\documentclass[12pt, journal, onecolumn]{IEEEtran}
\usepackage{graphicx,cite,epsfig,amssymb,amsmath}
\renewcommand{\QED}{\QEDopen}

\newtheorem{proposition}{Proposition}
\newtheorem{lemma}{Lemma}
\newtheorem{corollary}{Corollary}

\begin{document}
\title{\mbox{}\vspace{1.00cm}\\
\textsc{Capacity and Delay Tradeoff of Secondary Cellular Networks with Spectrum Aggregation}}
\author{\normalsize
Lingyu Chen$^1$, Chen Liu$^1$, Xuemin Hong$^1$, Cheng-Xiang Wang$^2$ \textit{Fellow, IEEE}, John Thompson$^3$, \textit{Fellow, IEEE}, and Jianghong Shi$^1$ \\
\vspace{1cm}
$^1$Department of Communications Engineering, \\
School of Information Science and Technology, \\
Xiamen University, Xiamen 361005, Fujian, P.R.China.\\
Email: \{chenly, chenl, xuemin.hong, shijh\}@xmu.edu.cn\\
\vspace{0.6cm}
$^2$Institute of Sensors, Signals and Systems, \\
School of Engineering and Physical Sciences, \\
Heriot-Watt University, Edinburgh, EH14 4AS, UK.\\
Email: cheng-xiang.wang@hw.ac.uk\\
\vspace{0.6cm}
$^3$Institute for Digital Communications,\\
School of Engineering,\\
University of Edinburgh, Edinburgh, EH9 3JL, UK.\\
Email: john.thompson@ed.ac.uk\\
\vspace{0.5cm}
\thanks{Submitted to IEEE Transaction on Wireless Communications.}
}
\maketitle
\newpage
\setcounter{page}{1}

\begin{abstract}


Cellular communication networks are plagued with redundant capacity, which results in low utilization and cost-effectiveness of network capital investments. The redundant capacity can be exploited to deliver secondary traffic that is ultra-elastic and delay-tolerant. In this paper, we propose an analytical framework to study the capacity-delay tradeoff of elastic/secondary traffic in large scale cellular networks with spectrum aggregation. Our framework integrates stochastic geometry and queueing theory models and gives analytical insights into the capacity-delay performance in the interference limited regime. Closed-form results are obtained to characterize the mean delay and delay distribution as functions of per user throughput capacity. The impacts of spectrum aggregation, user and base station (BS) densities, traffic session payload, and primary traffic dynamics on the capacity-delay tradeoff relationship are investigated. The fundamental capacity limit is derived and its scaling behavior is revealed. Our analysis shows the feasibility of providing secondary communication services over cellular networks and highlights some critical design issues.

\end{abstract}


\section{Introduction}
The capacity of a cellular radio access network (RAN) is fundamentally limited by the density of base stations (BSs), system bandwidth, and spectrum efficiency. Once a particular network is rolled out, its maximum capacity is relatively stable. The traffic load, on the other hand, changes dynamically across space and time. Because the capacity of a cellular network is planned to accommodate peak traffic demand, redundant capacity is unavoidable due to traffic fluctuations. Measurements campaigns (e.g., \cite{Mea14}) have shown that redundant capacity is a pervasive problem, which results in low utilization and cost-effectiveness of network capital investments.

Measurement also revealed that the major cause of mobile traffic is multi-media consumption\cite{Cisco14}, which includes different types of communication services. The first type is streaming services that are delay-sensitive but loss-tolerant. Typical applications include voice over IP and video conferencing. The second type is elastic traffic services that are delay-tolerant but loss-sensitive. Typical applications include web browsing and file transfer. In practice, there is no crucial difference in the delay constraints of streaming and elastic traffic. However, the emergence of new applications such as proactive caching \cite{ SZhou15, XWang14, NN12} brings a third type of traffic that has crucial difference from the first two types. Proactive caching systems are able to push content and cache them closer to end users, taking advantage of the fact that content demand is predictable and large cache space is becoming affordable. The traffic generated by proactive caching is called ultra-elastic traffic as the delay constraint is very relaxed (because user request happens much later) and the traffic demand is very flexible (because caching is transparent and opportunistic). It has been envisioned that the redundant capacity in cellular networks can be exploited to deliver ultra-elastic traffic as secondary traffic, which coexists with other higher priority traffic in the same cellular network \cite{Hong16}. Such a background motivates our study in this paper to investigate the fundamental capacity-delay tradeoff of secondary traffic in cellular networks.

Capacity is a key performance metric for cellular networks. It is widely believed that spectrum aggregation and cell densification will become prominent features of future cellular networks \cite{WangCX1, MaoGQ1}, leading to a heterogeneous cellular network (HCN) across multiple separated spectrum bands \cite{HSJ12}. The capacity and coverage performance of HCN has been studied extensively in the literature \cite{ HSJ12, HSD12, WCC12, MDR13, AG13, YS13, XL13, AK14, ZZ14, GN14, ZXM14, HSD14, ND15} using stochastic geometry. However, most of these studies focus on the physical layer capacity (i.e., spectrum efficiency), which differs from the throughput capacity considered in this paper. The throughput capacity of a user is evaluated with respect to the traffic queueing process of the user, and is defined as the average traffic arrival density (bits/s) that can be accommodated with finite delay. Throughput capacity addresses the higher layer aspects of traffic queueing and is better suited for cross-layer and delay related studies.

Delay is another key performance metric for analyzing cellular network traffic. Traffic behaviors can be modeled at the packet level or the session/flow level. Packet level dynamics are notoriously complicated as the temporal statistics exhibit self-similarity and multi-fractal behavior \cite{AF99}. On the contrary, session level behaviors justify the convenient assumption of Poisson arrival and can better reflect how the traffic performance is perceived by end users \cite{SB01}. Multiple modeling frameworks have been used for delay analysis in cellular systems with spectrum aggregation. The classic framework is queuing theory, which is theoretically mature but inflexible \cite{ QL05, LL06}. Discrete/continuous time Markov chain models \cite{ LJ12, VT12, JW13, VT11, MR09, NS15, SL14} are more flexible for describing some practical aspects, but fall short in providing closed-form analytical insights. The concept of local delay was proposed in \cite{MH13, ZG13}, but it focuses on the average delay and reveals no information about the delay distribution. In this paper, we adapt queueing theory as the framework to study the session level behavior of secondary traffic.

It is well-known that there is a fundamental tradeoff between capacity and delay. For mobile ad hoc networks, the capacity-delay tradeoff has been studied extensively using the framework of scaling law analysis \cite{ NeeM05, GamA06, GamA062, Li09, MaoGQ2}. The methodology and results therein, however, is not applicable to cellular networks. The capacity-delay tradeoff in cellular networks is still under-investigated due to the lack of a well-established analytical framework. Some early attempts \cite{AJ12, IS12, AF13} use interference approximation techniques to bound the session level performance of multi-cell networks. However, these works are not compatible with popular stochastic geometry models and only provide loose bounds for the estimation of mean delays. A framework of timely throughput was proposed in \cite{IH09} and recently adopted for the analysis of HCN in \cite{SL13,GZ16}. It assumes that a queuing packet will be dropped if the packet passes a critical delay. This is a useful framework that offers analytical tractability and flexibility in taking into account aspects such as mobility, user association bias, etc. However, this framework is better suited for the study of streaming traffic instead of elastic traffic, and does not provide a detailed characterization of delay distribution. Recent attempts to integrate stochastic geometry and queueing models was reported in \cite{BBa15, Hong1, Hong2}. In \cite{BBa15}, the spatial-temporal dependence of a cellular system is captured by some cell-load equations and eventually resolved via static simulations. Although this framework is mathematically rigorous, it lacks the analytical tractability to reveal closed-form insights. In our previous work \cite{Hong1, Hong2}, stochastic geometric and queueing models are combined to study the uplink capacity of hybrid ad-hoc networks with user collaboration. However, these work focused on a different type of network and did not fully address the issue of multi-user access, which is a critical feature of cellular systems. Moreover, \cite{Hong1, Hong2} did not explain how spatial-temporal dependencies in cellular networks can be decoupled to offer analytical insights. To our best knowledge, full integration of stochastic geometry and queueing theory in cellular networks is an unsolved and challenging problem due to complex coupling of network behavior in spatial and temporal domains \cite{BBa15}.

This paper proposes a new framework for the study of capacity-delay tradeoff of secondary/elastic traffic in cellular networks. Our framework integrates stochastic geometry and priority queuing models to offer analytical tractability and the ability to pinpoint delay distributions. The analytical tractability of our framework comes from two aspects. First, instead of trying to work out the exact mapping between spatial and temporal domains, our methodology is to identify a set of critical parameters in both domains and establish the relationships among their first-moment measures (mean values). Second, by focusing on the secondary traffic, we are able to justify certain assumptions and approximations. Our framework is shown to be useful in revealing some analytical insights. Specifically, this paper makes the following contributions.
\begin{itemize}
  \item Analytical results are derived to characterize the mean delay and delay distribution as functions of per user throughput capacity.
  \item Analytical results are derived to characterize the fundamental capacity limit in some special cases.
  \item A concise analytical approximation is obtained to describe how the per user capacity scales with user-BS density ratio.
\end{itemize}

The remainder of this paper is organized as follows. Section II describes the system model. The overall methodology and some useful approximations are introduced in Section III. The capacity-delay tradeoff and fundamental capacity limit are studied in Sections IV and V, respectively. Section VI provides numerical results and discussions. Finally, conclusions are drawn in Section VII.

\section{System Model}
\subsection{Secondary access protocol}
We consider the downlink of a large scale cellular network that aggregates $N$ independent frequency bands. BSs operating in the same band are assumed to have homogeneous bandwidth and transmit power denoted by $W_n$ and $P_n$, respectively, where $n$ ($1 \leq n \leq N$) is the index of bands. A user can operate in one band at a time, but can handover between different bands. Secondary users are assumed to comply with the following access protocol.
\begin{itemize}
  \item Step 1: Periodically check the buffer of secondary traffic. If the buffer is empty, remain in idle mode. Otherwise turn into active mode and proceed to Step 2.
  \item Step 2: Randomly select one band and associate with the nearest BS operating in the chosen band. This implies a Poisson-Voronoi cell model, which is a widely used cellular network model.
  \item Step 3: Evaluate whether the associated BS is vacant (i.e., not occupied by primary traffic) and available for secondary services. If yes, proceed to Step 4, otherwise return to Step 2. The probability that a typical BS in the $n$th band is vacant is called ``vacant probability'' and is denoted by $\Omega_n$. This parameter indicates the average load of primary traffic.
  \item Step 4: Evaluate the link quality with respect to the associated BS. If the signal-to-noise-and-interference ratio (SINR) is large enough to support a transmission rate of $R$ bits/s, proceed to Step 5, otherwise return to Step 2. Here, $R$ is the minimum rate requirement of secondary transmission. Such a requirement is imposed because it is desirable to restrict secondary services only to users with high quality links, otherwise secondary services may become inefficient and will result in excessive energy consumption and interference. The probability that a typical user in the $n$th band has good link quality is called ``coverage probability'' and denoted by $p_{n,v}$.
  \item Step 5: Compete with other in-coverage users for multiple access to the same BS. We assume a time-division multiple access (TDMA) scheme for multi-user access, where a band is fully allocated to one user at a time and multiple contending users have equal opportunities to access the band through time sharing. If contention is successful, proceed to Step 6, otherwise return to Step 2. The probability that a in-coverage secondary user in the $n$th tier is granted access is called ``access probability'' and denoted by $p_{n, a}$.
  \item Step 6: Transmit secondary traffic with a fixed rate $R$ until the buffer is empty. If the buffer is empty, proceed to Step 1. Otherwise if an outage (caused by primary traffic interruption or coverage outage) occurs during transmission, return to Step 2. We note that the transmission rate of secondary service is fixed to $R$ for simplicity.
\end{itemize}
For a user to receive secondary service in the $n$th band, he should firstly be associated with a vacant BS, secondly have a good coverage, and finally be granted access after multi-user contention. It follows that the service probability $\varepsilon_n$ is the product of vacant probability $\Omega_n$, coverage probability $p_{n,v}$, and access probability $p_{n,a}$, i.e.,
\begin{equation}
\varepsilon_n = \Omega_n \cdot p_{n,v} \cdot p_{n,a}. \label{epin}
\end{equation}

\subsection{Spatial interference model}
The spatial layout of BSs operating in the $n$th band is modeled by a stationary Poisson Point Process (PPP) in $\mathbf{R}^2$ with intensity $\lambda_{b,n}$, which is a commonly used model in the literature. For analytical tractability, we ignore the case of co-located BSs and assume that the spatial layout of BSs in different bands are independent. The spatial distribution of secondary users are also assumed to follow a stationary PPP in $\mathbf{R}^2$ with intensity $\lambda_{u}$. Let us consider a typical user in the $n$th band, the downlink SINR is a random variable, whose cumulative density function (CDF) has been derived for different types of fading channels \cite{And11}. For purposes of clarity and tractability, we consider a representative case in which the path loss exponent is 4. The complementary CDF of the user SINR is then given by \cite{And11}
\begin{equation}
F_{\gamma,n}(x) = \frac{\pi^{\frac{3}{2}} \lambda_{b,n}}{\sqrt{x / P_n}}e^{\frac{a^{2}}{\sqrt{2b}}}\mathrm{Q}\left( \frac{a}{\sqrt{2 x/P_n}}\right)  \label{CCDFapp}
\end{equation}
where $\mathrm{Q}(\cdot)$ denotes the $Q$-function and
\begin{equation}
a = \lambda_{b,n} \pi \left[ 1+\sqrt{x} \arctan(\sqrt{x}) \right].
\end{equation}
If the system is interference limited, which implies that $P_n$ is sufficiently large and the noise is negligible, (\ref{CCDFapp}) can be further simplified to \cite{And11}
\begin{equation}
F_{\gamma}^{\lim}(x) = \frac{1}{1 + \sqrt{x} \arctan(\sqrt{x}) }.  \label{CCDFlim}
\end{equation}

According to the secondary access protocol, a user is in coverage of secondary services if $W_n \log_2(1+\gamma_n) \geq R$, where $\gamma_n$ denotes the SINR perceived by the user. The coverage probability in the $n$th band is therefore given by
\begin{equation}
p_{n,v} = F_{\gamma, n}(2^{R/W_n}-1). \label{pnv}
\end{equation}


\begin{figure}[t]
\vspace{0.125in}
\centerline{\includegraphics[width=11cm,draft=false]{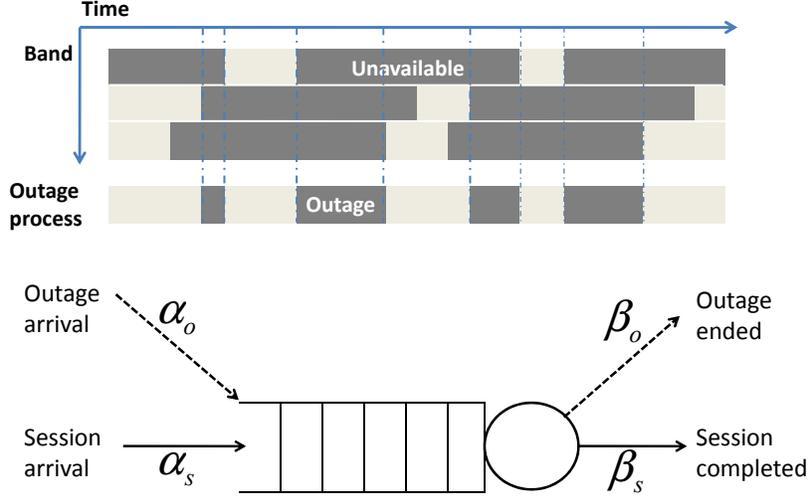}}
\caption{\small Priority queuing model of a typical user with secondary traffic and random outage.}
\end{figure}

\subsection{Temporal queuing model}
As illustrated in Fig.~1, we model the secondary traffic dynamic as a preemptive priority queue, where the transmission of secondary traffic may be preempted (i.e., immediately interrupted) by outages. An outage can be caused by multiple factors such as primary traffic interruption, bad coverage, and failure in multi-user contention. We assume users can handover between bands with negligible time, hence an outage only occurs when no band is available for secondary services. We propose to model the composite outage effect as a stream of higher priority traffic in the priority queue. The arrival of outage events follows a Poisson process with mean interval $\bar{\alpha}_o$. Each outage event contributes to an additive random outage duration denoted by $\beta_o$, the mean of which is $\bar{\beta_o}$. Let us define
\begin{equation}
\rho_o = \bar{\beta_o}/\bar{\alpha_o}.
\end{equation}
This parameter represents the fraction of time that a user is in outage and cannot be served by a BS in all bands. It is worth noting that we do not make any particular assumption on the distribution of $\beta_o$, i.e., it can follow an arbitrary form of continuous distribution. This gives our model the flexibility to represent a wide range of outage phenomenons.

We consider the secondary traffic behavior at the session level. Users are assumed to have homogeneous incoming traffic of sessions that follow i.i.d. Poisson arrival process with mean interval $\bar{\alpha}_s$. Each session carries a file of random size $L$ to be delivered from the BS to the user. The file size $L$ follows a general distribution with mean $\bar{L}$. The mean throughput capacity of a user is given by
\begin{equation}
C=\bar{L}/\bar{\alpha}_s.
\end{equation}
Under the assumption of constant transmission rate $R$, the transmission time of a session is a random variable $\beta_s = L/R$. Let us define
\begin{equation}
\rho_s = \bar{\beta}_s/\bar{\alpha}_s = \bar{L}/(R \bar{\alpha}_s) = C/R.
\end{equation}
This parameter represents the fraction of time that a user receives transmission from a BS. The file size $L$ is assumed to follow a general distribution.

The transmission of a secondary session is forced to stop immediately once an outage occurs. Once the secondary service is available again, a session may adopt a `resume' policy to transmit from where it stopped, or adopt a `repeat' policy to retransmit from the beginning. Our paper is restricted to the resume policy, noting that an extension to the repeat policy is straightforward. Based on the above modeling assumptions, the queuing process at a typical secondary user can be readily captured by a M/G/1 two-level priority queuing model with a preemptive resume policy \cite{Queuebook}. Such a classic queuing model is fully characterized by the four random variables shown in Fig.~1. 

\section{Methodology and Approximations}
Our system model describes a large scale, dynamic system in the spatial and temporal domains. These two domains are inherently coupled and correlated. Analysis encompassing both domains requires integration of stochastic geometry and queueing models, which is an extremely challenging task. Existing work resorted to static simulation to yield results without revealing much theoretical insight \cite{BBa15}. In this paper, instead of trying to capture the detailed relationships between the spatial and temporal domains, we propose a methodology that connects these two domains by establishing analytical relationships among the first-order statistic measure (i.e., mean values) of some critical parameters. Higher order statistics are not our primary concern, therefore we use the flexible model of M/G/1 queue (with general distributions) to offer sufficient flexibility to represent a wide range of higher order statistics. This section will first explain our overall methodology and then introduce some useful approximations as preliminaries.

\subsection{Overall methodology}
Fig.~2 illustrates our overall approach to address the connections between spatial and temporal domains. Our analysis implies two underlying assumptions. First, the queueing processes of users are assumed to be independent and homogeneous. This assumption is reasonable because in the macro time-scale, users are assumed to have independent mobility traces; while in the micro time-scale, users are allowed to hop randomly between independent bands. The composite effects of random mobility and band selection renders the queueing process of a user to be independent of others in the long term. In this case, we can consider a typical user with a typical queueing process, at a typical location and associated with a typical BS. A typical user can be understood as an arbitrary user or a randomly selected user. A probability space can also be defined for the typical user for its states of queueing and signal reception, etc. The second assumption is that all BSs constantly transmit with power $P_n$. This assumption decouples the interference statistics with user behavior and represents the worst-case interfering scenario. It is reasonable because the combined load of primary and secondary traffic from multiple users is likely to keep BSs in constant transmission.

\begin{figure}[t]
\vspace{0.125in}
\centerline{\includegraphics[width=11cm,draft=false]{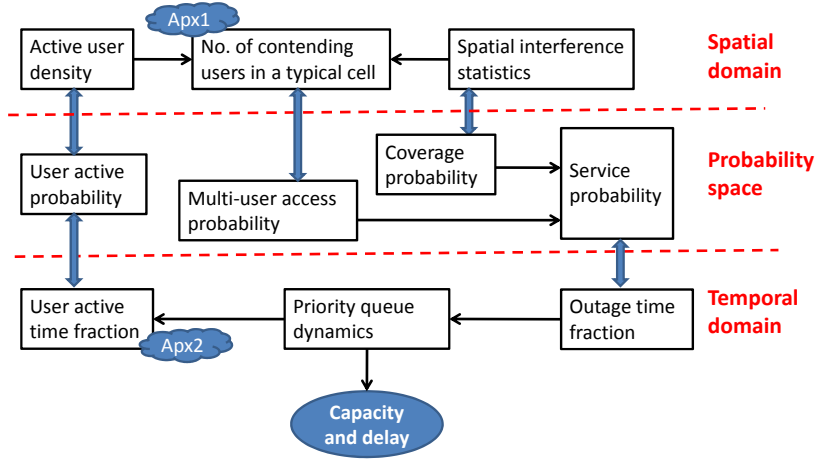}}
\caption{\small Connections among parameters in the spatial and temporal domains.}
\end{figure}

According to the ergodic theory, when the queueing process of the typical user has a statistical equilibrium, the queueing process is ergodic \cite{Queuebook} and hence the time average of a queueing parameter is identical to the average over the probability space. This allows us to map time-domain parameters to the probability space. Moreover, according to the theory of Palm probability in stochastic geometry, the spatial average of a large scale network is identical to the probabilistic average over the typical user/BS \cite{Geometry}. This allows us to map spatial-domain parameters to the probability space. Based on these mappings, we are able to deduce a chain of relations in Fig.~2 as follows.

Let us consider the \emph{outage time fraction} in a typical queue, which is the average fraction of time that secondary services is not available. The \emph{outage time fraction} affects the queueing dynamics and hence the \emph{user active time fraction}, which is the average fraction of time that there is secondary traffic buffered in the queue. The \emph{user active time fraction} is identical to the \emph{active probability} of a typical user, which affects the \emph{active user density} in the spatial domain. \emph{Active user density} and \emph{spatial interference statistics} both affect the distribution of the \emph{number of contending users in a typical cell}, which determines the \emph{multi-user access probability}. \emph{Spatial interference statistics} also affects the \emph{coverage probability} of a typical user. Moreover, as shown in (\ref{epin}), the \emph{access probability} and \emph{coverage probability} affects the \emph{service probability}, which ultimately determines the \emph{outage time fraction}. In other words, we have
\begin{equation}
\varepsilon = 1 - \rho_o
\end{equation}
where $\varepsilon$ is the service probability, $\rho_o$ is the outage time fraction, and $\rho_o$ can be expressed as a function of $\varepsilon$. The above chain of relations allows us to establish an equilibrium equation that connects first-order statistics of multiple parameters in the spatial and temporal domains. To establish the equation in an analytical form, two approximations are further introduced.


\subsection{Approximation to the number of in-coverage users in a typical cell}
The probability density function (PDF) of the size of a typical Poisson Voronoi cell is analytically intractable but can be approximated using the Monte Carlo method. Let $\lambda$ be the density of the underlying Poisson process and $V$ denote the random size of a typical Voronoi cell normalized by $1/\lambda$. The PDF of $V$ is given by \cite{JSF07}
\begin{equation}
f_V(x) = \frac{3.5^{3.5}}{\Gamma(3.5)} x^{2.5} e^{-3.5x}
\end{equation}
where $\Gamma(\cdot)$ is the gamma function. Moreover, consider an arbitrary user and the random size $U$ of the Voronoi cell to which the user belongs to. The PDF of $U$ normalized by $1/\lambda$ is given by \cite{Kim12}
\begin{equation}
f_U(x) = \frac{3.5^{4.5}}{\Gamma(4.5)} x^{3.5} e^{-3.5x}.
\end{equation}
The difference between $f_V(x)$ and $f_U(x)$ comes from the fact that a user has a higher chance to be covered by larger Voronoi cells.

Let us consider a single band of the network with BS density $\lambda_b$ and user density $\lambda_u$. Denoting $K_1$ as the random number of users in a non-empty Voronoi cell, the probability mass function (PMF) of $K_1$ is given by
\begin{equation}
f_{K_1}(k) = \int_0^{\infty} \frac{( \frac{\lambda_u}{\lambda_b} x)^k}{k!} e^{-\frac{\lambda_u}{\lambda_b} x} f_U(x) dx. \label{fK1}
\end{equation}
Let $K$ be the random number of `in-coverage' users in a Voronoi cell. The distribution of $K$ is related to the size and shape of the cell and it is difficult to obtain its exact PMF. Keeping the basic form of (\ref{fK1}), we propose an approximation to the PMF of $K$ given by
\begin{align}
f_{K}(k) \approx & \int_0^{\infty} \frac{( p \Lambda \frac{\lambda_u}{\lambda_b} x)^k}{k!} e^{- p \Lambda \frac{\lambda_u}{\lambda_b} x} f_U(x) dx \label{fK} \\
         =& \frac{3.5^{4.5} \Gamma(4.5 + k)}{\Gamma(4.5) k!} \frac{(\Lambda \lambda_u p / \lambda_b)^k}{(3.5 + \Lambda \lambda_u p / \lambda_b)^{4.5+k}}. \nonumber \label{fK}
\end{align}
where the parameters $p$ and $\Lambda$ are introduced to capture the effect of colored thinning on the original user point process. Here, $p$ is the probability that an arbitrary user falls within coverage (with target rate $R$) and can be calculated by (\ref{pnv}). The coefficient $\Lambda$ is an artificial constant to capture the effect of colored thinning. The value of $\Lambda$ is obtained by searching for the best fit of (\ref{fK}) to the empirical PMF obtained via Monte Carlo simulations. Through extensive simulations we find that given $\Lambda = 2/3$, the approximation in (\ref{fK}) is valid for a wide range of practical values for $\lambda_u$ and $\lambda_b$. Fig.~3 illustrates the accuracy of this approximation.

\begin{figure}[t]
\vspace{0.125in}
\centerline{\includegraphics[width=11cm,draft=false]{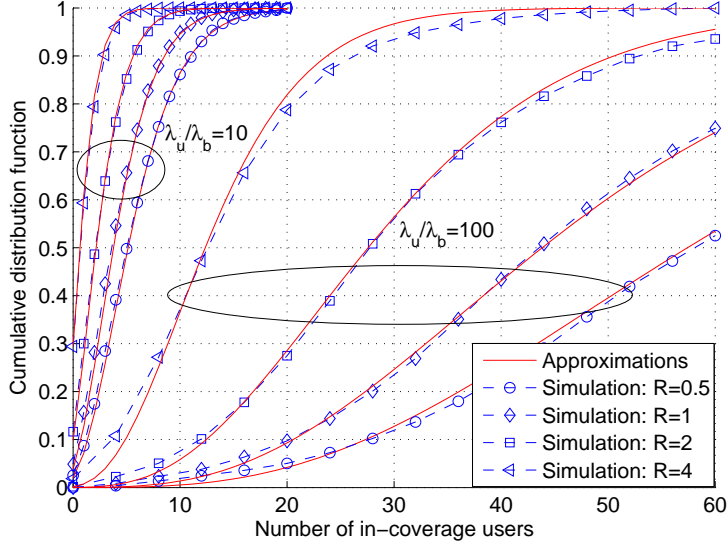}}
\caption{\small Approximation on the probability density function of in-coverage users in a typical cell ($\lambda_b=10^{-6}$).}
\end{figure}


\subsection{Approximation to user active time fraction}
A user is active when there are sessions buffered or being transmitted in the queue. We are interested in the probability $p_{active}$ that a typical user stays active. This probability also represents the fraction of time for a user to be active. Let $T$ be the random transmission time of a session. The mean value of $T$ is given by \cite{Queuebook}
\begin{equation}
\bar{T} = \frac{\bar{\beta}_s}{1-\rho_o}.
\end{equation}
The exact PDF of $T$ is not exponential, but for the purpose of calculating the user active probability, we assume that $T$ follows an exponential distribution with mean $\bar{T}$. The accuracy of this approximation is illustrated in Fig.~4, where we assume exponentially distributed $\beta_o$ and $\beta_s$, set $\bar{\alpha}_o=0.1$, $\bar{\alpha}_s=1$, $\varepsilon_o=0.3$, and let $\varepsilon_s$ varies from 0.1 to 0.5. The exact PDF of $T$ is obtained from its Laplace transform according to (\ref{zii1}). We find that the exponential approximation is valid under the condition that the arrival rate of outage is greater than the arrival rate of secondary traffic session. This condition is realistic because our system model considers delay at the session level, which has a larger time scale than outages caused by packet-level primary traffic.

\begin{figure}[t]
\vspace{0.125in}
\centerline{\includegraphics[width=11cm,draft=false]{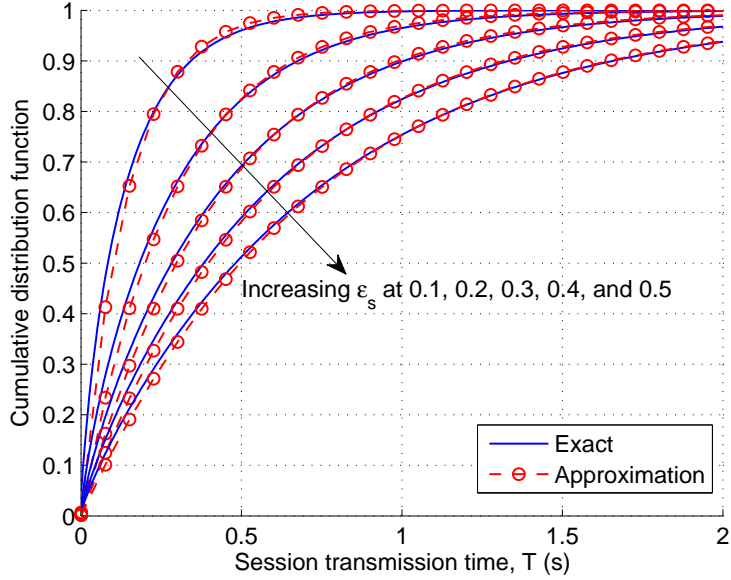}}
\caption{\small Exponential approximation for the CDF of session transmission time $T$ ($\bar{\alpha}_o=0.1$, $\bar{\alpha}_s=1$, $\varepsilon_o=0.3$).}
\end{figure}

Now let us consider a discrete-value stochastic process representing the number of sessions staying in the queue. Based on the above mentioned exponential approximation, it is easy to see that this process is a classic birth-death process \cite{Queuebook} characterized by an uniform birth rate $1/\bar{\alpha}_s$ and death rate $1/\bar{T}$. Let $\phi_k$ (k = 0,1,2,3...) denote the steady state probability that there are $k$ sessions in the queue. The equilibrium condition of the birth-death process gives $\phi_k = (\bar{T}/\bar{\alpha}_s)^k \phi_0$. By further considering the constraint of total probability $\Sigma_{k=0}^{\infty} \phi_k = 1$, we have $\phi_0 = 1 - \bar{T}/\bar{\alpha}_s$. It follows that
\begin{equation}
p_{active} = 1 - \phi_0 = \bar{T}/\bar{\alpha}_s = \frac{\bar{\beta}_s}{(1 - \rho_o) \bar{\alpha}_s } = \frac{\rho_s}{1 - \rho_o} = \frac{\rho_s}{\varepsilon}. \label{pactive}
\end{equation}

\section{Capacity-Delay Tradeoff Analysis}
\subsection{Useful Lemmas}

\begin{lemma}
Let $\varepsilon_n$ denote the probability that a user can be served with a target rate $R$ by the nearest BS operating in the $n$th band. We have
\begin{equation}
\varepsilon_n = \frac{\Omega_n}{\Lambda \lambda_n} \left[ 1- \left( 1 + \frac{\Lambda p_{n,v} \lambda_n}{3.5} \right)^{-3.5} \right]
\end{equation}
where $p_{n,v}$ is given by (\ref{pnv}), $\Lambda=2/3$, and
\begin{equation}
\lambda_n = \frac{\lambda_u}{\lambda_{b,n}} \cdot \frac{\rho_s}{\varepsilon} \cdot \frac{\Omega_n p_{n,v}}{\sum_{n=1}^N \Omega_n p_{n,v}}.
\end{equation}
\end{lemma}

\begin{proof}
See Appendix I.
\end{proof}

\begin{lemma}
When the cellular system independently operates a total number of $N$ bands, the probability that a user can be served by at least one band with a targeted rate $R$ is
\begin{equation}
\varepsilon = 1 - \prod_{n=1}^{N}(1 - \varepsilon_n). \label{vare}
\end{equation}
\end{lemma}
\begin{proof}
It is straightforward to see that ($1 - \varepsilon$) equals the joint probability that all bands fail to provide service to a user with the target rate $R$.
\end{proof}
According to Lemma 1, $\varepsilon_n$ is itself a function of $\varepsilon$. Therefore Lemma 2 gives a non-linear equation of $\varepsilon$, based on which the value of $\varepsilon$ can be calculated by solving the non-linear equation via numerical methods. In the special case that all bands have the same characteristics in terms of bandwidth, transmit power, and availability, (\ref{vare}) can be simplified to
\begin{equation}
\varepsilon_N = 1 - (1 - \varepsilon_n)^N.
\end{equation}
In case when $N=1$, (\ref{vare}) can be solved to give $\varepsilon$ as an explicit function related to capacity $C$ and target rate $R$ as follows
\begin{equation}
\varepsilon_1 = \frac{p_{n,v}}{3.5} \frac{\Lambda \lambda_u}{\lambda_{b,n}} \frac{C}{R} \left[ 1 - \left(1 - \frac{\Lambda \lambda_u C}{\Omega_n \lambda_{b,n} R} \right)^{-2/7} \right]^{-1}.
\end{equation}

\subsection{General results for capacity-delay tradeoff}
Once the value of $\varepsilon$ is obtained, we can evaluate the mean delay and delay distribution of a session. Established results for two-class M/G/1 priority queues with preemptive-resume policy \cite{Queuebook} can be directly applied to give the following two propositions.

\begin{proposition}
The mean delay of a session is given by
\begin{equation}
\bar{D} = \frac{1}{2 \varepsilon (\varepsilon - \frac{C}{R})} \left( \frac{\hat{\beta}_s}{\bar{\alpha}_s} + \frac{\hat{\beta}_o}{\bar{\alpha}_o} \right) + \frac{\bar{L}}{R \varepsilon}  \label{Dgen}
\end{equation}
where $\hat{\beta}_s$ and $\hat{\beta}_o$ are the second-order moments of random variables $\beta_s$ and $\beta_o$, respectively.
\end{proposition}

The delay of a session is the total time the session spends in the queue and consists of two parts. The first part is waiting time $W$, which is the duration from the moment of arrival to the moment when the transmission starts. The second part is transmission time $T$, which is the duration from the moment when transmission starts to the moment when the transmission ends. It follows that $D = W + T$, where $W$ and $T$ are independent RVs \cite{Queuebook}. The PDF of $D$ cannot be obtained directly. However, the Laplace transforms of the PDFs of $W$ and $T$ can be evaluated. Let $\mathfrak{L}_X(\cdot)$ denote the Laplace transform to the PDF of random variable $X$, we have the following proposition.
\begin{proposition}
The Laplace transform of the random delay $D$ of a typical session is given by
\begin{equation}
\mathfrak{L}_D(s) = \mathfrak{L}_T(s) \mathfrak{L}_W(s). \label{lds}
\end{equation}
Here, $\mathfrak{L}_T(s)$ is given by
\begin{equation}
\mathfrak{L}_T(s) = \mathfrak{L}_{\beta_{\rm{s}}} \left[ K \left( s \right) \right] \label{zii1}
\end{equation}
where  $\mathfrak{L}_{\beta_{s}}(\cdot)$ is the Laplace transform of $\beta_s$ and
\begin{equation}
K\left( s \right) = s + \frac{1 - G\left( s \right)}{\bar{\alpha}_o}. \label{zii2}
\end{equation}
Here, $G(s)$ is the solution with the smallest absolute value that satisfies the following equation
\begin{equation}
x - {\mathfrak{L}_{{\beta_o}}}\left( {s + \frac{1 - x}{\bar{\alpha}_o}} \right) = 0 \label{zii3}
\end{equation}
where $\mathfrak{L}_{\beta_{o}}(\cdot)$ is the Laplace transform of $\beta_o$.
The second term $\mathfrak{L}_W(s)$ in (\ref{lds}) is given by
\begin{equation}
\mathfrak{L}_W(s) = (1 - \rho_o - \rho_s) {\bar{\alpha}_s}{{K\left( s \right)} \over {{\mathfrak{L}_{{\beta _s}}}\left[ {K\left( s \right)} \right] + {{\bar{\alpha}}_s}s - 1}}.   \label{wii}
\end{equation}
\end{proposition}

\subsection{Capacity-delay tradeoff in special cases}
\subsubsection{Exponential distribution}
Propositions 1 and 2 are applicable when both the file size $L$ and outage duration $\beta_o$ follow general distributions. In the special case where both $L$ and $\beta_o$ follow exponential distributions, we have $\hat{\beta}_s = 2(\bar{\beta}_s)^2$ and $\hat{\beta}_o = 2(\bar{\beta}_o)^2$. The mean delay becomes
\begin{equation}
\bar{D} = \frac{1}{ \varepsilon (\varepsilon - \frac{C}{R})} \left( \frac{C \bar{L}}{R^2} + (1 - \varepsilon)^2 \bar{\alpha}_o \right) + \frac{\bar{L}}{R \varepsilon}.  \label{Dexp}
\end{equation}
Moreover, given an exponential random variable $X \sim \exp(\bar{X})$, its Laplace transform can be evaluated as
\begin{equation}
\mathfrak{L}_{exp} (s) = {\frac{1}{1 + s \bar{X}}}. \label{Lx}
\end{equation}
Based on (\ref{Lx}), closed-form Laplace transforms of $\beta_s=L/R$ and $\beta_o$ can be obtained in (\ref{zii1}) and (\ref{zii3}). It follows that Eqn. (\ref{zii3}) can be solved explicitly to give
\begin{equation}
G\left( s \right) = \frac{\left( {1 + {\varepsilon_o} + s{{\bar{\beta}}_o}} \right) - \sqrt {{{\left( {1 + {\varepsilon_o} + s{{\bar{\beta}}_o}} \right)}^2} - 4{\varepsilon_o}}} {2 \varepsilon _o }.
\end{equation}
\subsubsection{Gamma distribution}
A more general distribution we can consider for $L$ and $\beta_o$ is  Gamma distribution, which provides more flexibility to model a variety of practical scenarios. The PDF of Gamma distribution is given by
\begin{equation}
\Gamma(k, \theta) = \frac{1}{\theta ^k} \frac{1}{\Gamma \left( k \right)} t^{k - 1} e^{- \frac{t}{\theta}}
\end{equation}
where $k$ and $\theta$ are the shape and scale parameters, respectively. The first and second moments of the Gamma distribution are $k \theta$ and $k (k+1) \theta^2$, respectively. Let $L \sim \Gamma(k_L, \bar{L}/k_L)$ and $\beta_o \sim \Gamma(k_{\beta_o}, \bar{\beta}_o/k_{\beta_o})$. Here we introduce two new parameters $k_L$ and $k_{\beta_o}$ to characterize the shape of distributions of $L$ and $\beta_o$, respectively. It follows that $\beta_s = L/R \sim \Gamma(k_L, \bar{L}/(k_L R))$, and the mean delay in (\ref{Dgen}) becomes
\begin{equation}
\bar{D} = \frac{1}{2 \varepsilon (\varepsilon - \frac{C}{R})} \left( \frac{C \bar{L}}{R^2} \frac{k_L+1}{k_L} + (1 - \varepsilon)^2 \bar{\alpha}_o \frac{k_{\beta_o}+1}{k_{\beta_o}} \right) + \frac{\bar{L}}{R \varepsilon}. \label{Dgamma}
\end{equation}
It is easy to see that when $k_L=1$ and $k_{\beta_o}=1$, the Gamma distribution is reduced to exponential distribution and (\ref{Dgamma}) is reduced to (\ref{Dexp}).

To evaluate the delay distribution, we have the Laplace transform of $G \sim \Gamma(k, \theta)$ given by
\begin{equation}
\mathfrak{L}_{gamma} \left( s \right) = {\left( {1 + \theta s} \right)^{- k}}. \label{Lg}
\end{equation}
Based on (\ref{Lg}), closed-form Laplace transforms of $\beta_s=L/R$ and $\beta_o$ can be obtained according to (\ref{zii1}) and (\ref{zii3}). It follows that when $k$ is an integer or a rational fraction, Eqn. (\ref{zii3}) yields a polynomial form. Therefore the function $G(s)$ in (\ref{zii3}) can be easily solved using existing root-finding algorithms for polynomials.

\section{Capacity Limit and Scaling}

This section studies the fundamental capacity limit at the interference limited regime and investigates how the capacity limit scales with bandwidth and user-BS density ratio. The capacity limit is defined as the maximum capacity that permits a stable queue at a typical user. It is also the capacity that gives infinite mean delay. Interference-limited regime means that power $P_n$ is sufficiently large to justify the closed-form SINR CCDF in (\ref{CCDFlim}). For simplicity, we assume that the $N$ bands have homogeneous characteristics in terms of bandwidth and BS density. Two different cases are considered. The first case assumes a fixed bandwidth of each band, which means the system bandwidth scales linearly with $N$. This case is useful when we want to investigate the impact of spectrum aggregation on the system capacity. The second case assumes a fixed system bandwidth, which means the bandwidth per band is inversely proportional to $N$. This case is relevant when we are interested in the impacts of spectrum sharing and channelization on the system capacity. Throughout this section, we use the capital letter `N' as the footnote of parameters to emphasize that we consider homogeneous bands. For example, $W_n$, $\varepsilon_n$ and $\lambda_n$ are replaced by $W_N$, $\varepsilon_N$ and $\lambda_N$, respectively.

\subsection{Fixed bandwidth per band}

\begin{proposition}
In the case of fixed bandwidth per band, the capacity limit $C^{\lim}_I$ is a function of $R$, $\lambda_u$, $\lambda_b$, and $N$ given by
\begin{equation}
C^{\lim}_I = R \left[ 1 - (1 - \varepsilon_N)^N  \right] \label{ClimI}
\end{equation}
where
\begin{equation}
\varepsilon_N = \frac{\Omega_N}{\Lambda \lambda_N}\left[ 1 - \left(1 + \frac{\Lambda \lambda_N p_N^I}{3.5} \right)^{-3.5}  \right].\label{epilimI}
\end{equation}
Here, $p_N^I$ is given by
\begin{equation}
p_N^I = \left( 1 + \sqrt{2^{R/W_N}-1} \arctan\sqrt{2^{R/W_N}-1} \right)^{-1} \label{plimI}
\end{equation}
and $\lambda_N = \lambda_u/(\lambda_{b,n} N)$.
\end{proposition}
\begin{proof}
A stable queue requires $1 - \rho_o - \rho_s >0$, which gives $\varepsilon > \rho_s = C/R$. The capacity limit is achieved when the equality holds, i.e., $\varepsilon = C/R$ or $\rho_d/(1-\rho_o)=1$. Substituting this equation into Lemma 1 yields $\varepsilon_N$ in (\ref{epilimI}).
\end{proof}

We note that by considering the limiting condition, $\varepsilon_N$ can be expressed as an explicit function of other parameters (as opposed to numerically solving a non-linear equation in Lemma 2). This allows us to express the capacity limit as a closed-form function of $R$, $N$, $\lambda_u$, and $\lambda_b$, as shown in (\ref{ClimI}). In the case of fixed bandwidth per band, we are interested in the following optimization problem: given $N$ and the network environment $\lambda_u$ and $\lambda_b$, how can we choose a proper target rate $R$ to maximize the capacity limit? This optimization problem can be formally stated as $C_I^{\max} = \max\limits_{R}(C^{\lim}_I)$. To better understand the nature of this optimization problem, representative numerical examples are presented in Fig.~5 to show $C^{\lim}_I$ as a function of $R$. We see that there is an unique maximum value of $C^{\lim}_I$, which is achieved when the first-order derivative $dC^{\lim}_I/dR$ equals zero. According to Proposition 3, the derivative function $dC^{\lim}_I/dR$ can be obtained in closed-form to give the following corollary.
\begin{corollary}
The optimum value $R$ for the optimization problem $C_I^{\max} = \max \limits_{R}(C^{\lim}_I)$ is given by the root of the following non-linear equation:
\begin{equation}
\frac{d C^{\lim}_I}{d R} = R f_o^{'}(R) + f_0(R) = 0
\end{equation}
where
\begin{align}
f_0(R) = & 1 - (1-\varepsilon_N)^N  \\
f_0^{'}(R) = & f_1(R) \cdot f_2(R) \cdot f_3(R) \cdot f_4(R) \\
f_1(R) = & N(1-\varepsilon_N)^{N-1} \\
f_2(R) = & \left( 1 + \frac{\lambda_N p_N^I}{3.5} \right)^{-4.5} \\
f_3(R) = & - \frac{\arctan(\chi_N) + (1+\chi_N^2)^(-1)}{( 1 + \chi_N \arctan(\chi_N))^2} \\
f_4(R) = & \frac{\ln 2}{2} 2^R \left( 2^R -1 \right)^{-1/2} \\
\chi_N = & \sqrt{2^R -1}.
\end{align}
In the above equations, $p_N^I$ is defined in (\ref{plimI}) and $\varepsilon_N$ is defined in (\ref{epilimI}).
\end{corollary}

Based on the above corollary, the first-order derivative function $dC^{\lim}_I/dR$ is calculated and shown in Fig.~5. The root is obtained by solving the non-linear equation and shown to be accurate for achieving the maximum value of $C^{\lim}_I$.

\begin{figure}[t]
\vspace{0.125in}
\centerline{\includegraphics[width=11cm,draft=false]{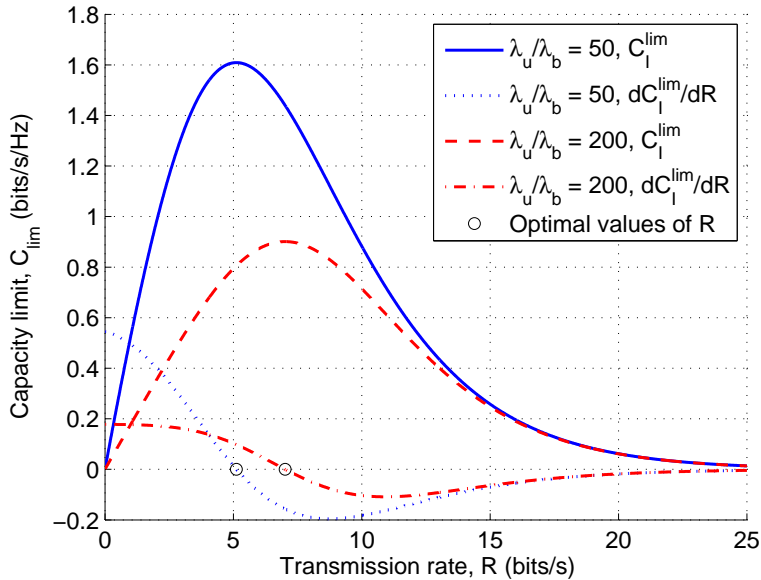}}
\caption{\small Capacity limit $C^{\lim}_{I}$ and its first-order derivative as a function of $R$ (fixed bandwidth per band, N=5).}
\end{figure}

\subsection{Fixed system bandwidth}
In this case, the total system bandwidth is normalized to 1 and the bandwidth of each band becomes $1/N$. Define the capacity limit $C^{\lim}_{II}$ as the maximum achievable capacity for a stable queue given $R$, $N$, $\lambda_u$, and $\lambda_b$. Further define the maximum capacity as $C_{II}^{\max} = \max\limits_{R}(C^{\lim}_{II})$. We have the following two propositions.

\begin{proposition}
The capacity limit $C^{\lim}_{II}$ can be calculated according to Proposition 3 by replacing $p_N^I$ with $p_N^{II}$, where
\begin{equation}
p_N^{II} = \left( 1 + \sqrt{2^{RN/W_N}-1} \arctan\sqrt{2^{RN/W_N}-1} \right)^{-1}. \label{plimII}
\end{equation}
\end{proposition}
\begin{proof}
The proof is straightforward by following the proof of Proposition 3 and setting the channel bandwidth to $1/N$.
\end{proof}

\begin{proposition}
The maximum capacity is given by
\begin{equation}
C^{\max}_{II} = C^{\max}_{I}/N
\end{equation}
 where $C^{\max}_{I}$ can be calculated from Corollary 1.
\end{proposition}
\begin{proof}
According to Propositions 3 and 4, we can write $C^{\lim}_{II}(R) = C^{\lim}_{I}(R N)/N$. Further considering the fact that adding a scaling on $R$ will not change the maximum value of $C^{\lim}_{I}$, i.e., $\max\limits_{R} C^{\lim}_{I}(R) = \max\limits_{R} C^{\lim}_{I}(R N) = C^{\max}_{I}$, Proposition 5 can be proved.
\end{proof}

\section{Numerical results and discussions}

This section presents numerical results and discusses their implications. First, we aim to understand the impacts of various parameters on the capacity-delay tradeoff (Figs.~6 to 10). Second, we want to investigate how the fundamental capacity limit scales with the number of bands $N$ and user-BS density ratio (Figs.~11 to 13). For illustration purpose, we consider an interference-limited system and homogeneous bands with $W_N=1$ and $\Omega_N =1$.

\subsection{Capacity-delay tradeoff}
Due to page limits, we restrict our discussions to the mean delay and the case of fixed bandwidth per band. Except when otherwise mentioned, the default parameter values are set to be $N=5$, $\lambda_u/\lambda_b = 50$, $\bar{L} = 10$, and $\bar{\alpha}_o = 10$. Moreover, the distributions of $L$ and $\alpha_o$ are treated as exponential. Therefore, our subsequent discussions are primarily based on Eqn. (\ref{Dexp}).

Fig.~6 shows the mean delay $\bar{D}$ as a function of $R$ with varying $N$ while the capacity is fixed to $C=1$ bits/s. U-shape curves are observed, indicating that given other parameters, there is an optimal value for $R$ to minimize the mean delay. Because we are interested in the fundamental capacity-delay tradeoff, it is desirable to consider the minimized delay over feasible values of $R$. Define $\bar{D}_{\min} = \min\limits_{R}(\bar{D})$, we will subsequently evaluate $\bar{D}_{\min}$ as a function of $C$. The value of $\bar{D}_{\min}$ is obtained by performing a numerical optimization over $R$.

Fig.~7 shows the impact of $\lambda_u/\lambda_b$ on the capacity-delay tradeoff curve. Two interesting phenomena are observed. First, when the user-BS density ratio is relatively high ($100 \leq \lambda_u/\lambda_b \leq 1000$), the capacity per user (at a fixed delay) appears to scale linearly with $\lambda_b/\lambda_u$. We called this ``infrastructure-limited'' regime, in which the investment in BS infrastructure yields linear returns on the capacity. In contrast, when the user-BS density is relatively low ($10 \leq \lambda_u/\lambda_b \leq 100$), investment in BS infrastructure only yields sub-linear returns.
Second, in the low delay regime, there is minimum delay even when $C$ approaches zero. Such a minimum delay is caused by coverage outage and primary traffic interruption, which caps the secondary service probability.

\begin{figure}[p]
\centerline{\includegraphics[width=11cm,draft=false]{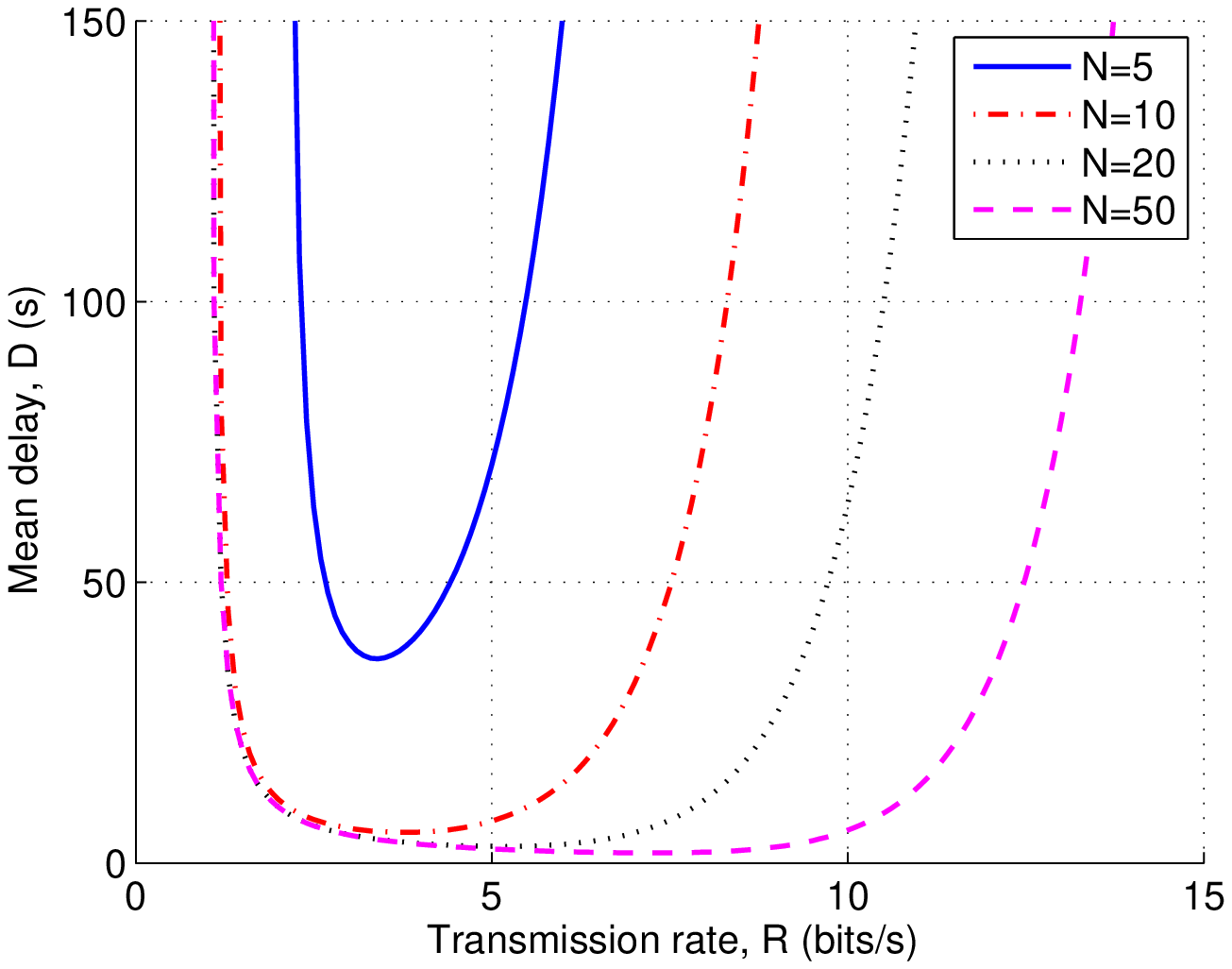}}
\caption{\small Mean delay $\bar{D}$ as a function of $R$ with varying $N$ (fixed bandwidth per band, $C$=1 bit/s/Hz, $\lambda_u/\lambda_b$=50, $\bar{L}$=10, $\bar{\alpha}_o$=10).}
\end{figure}

\begin{figure}[p]
\centerline{\includegraphics[width=11cm,draft=false]{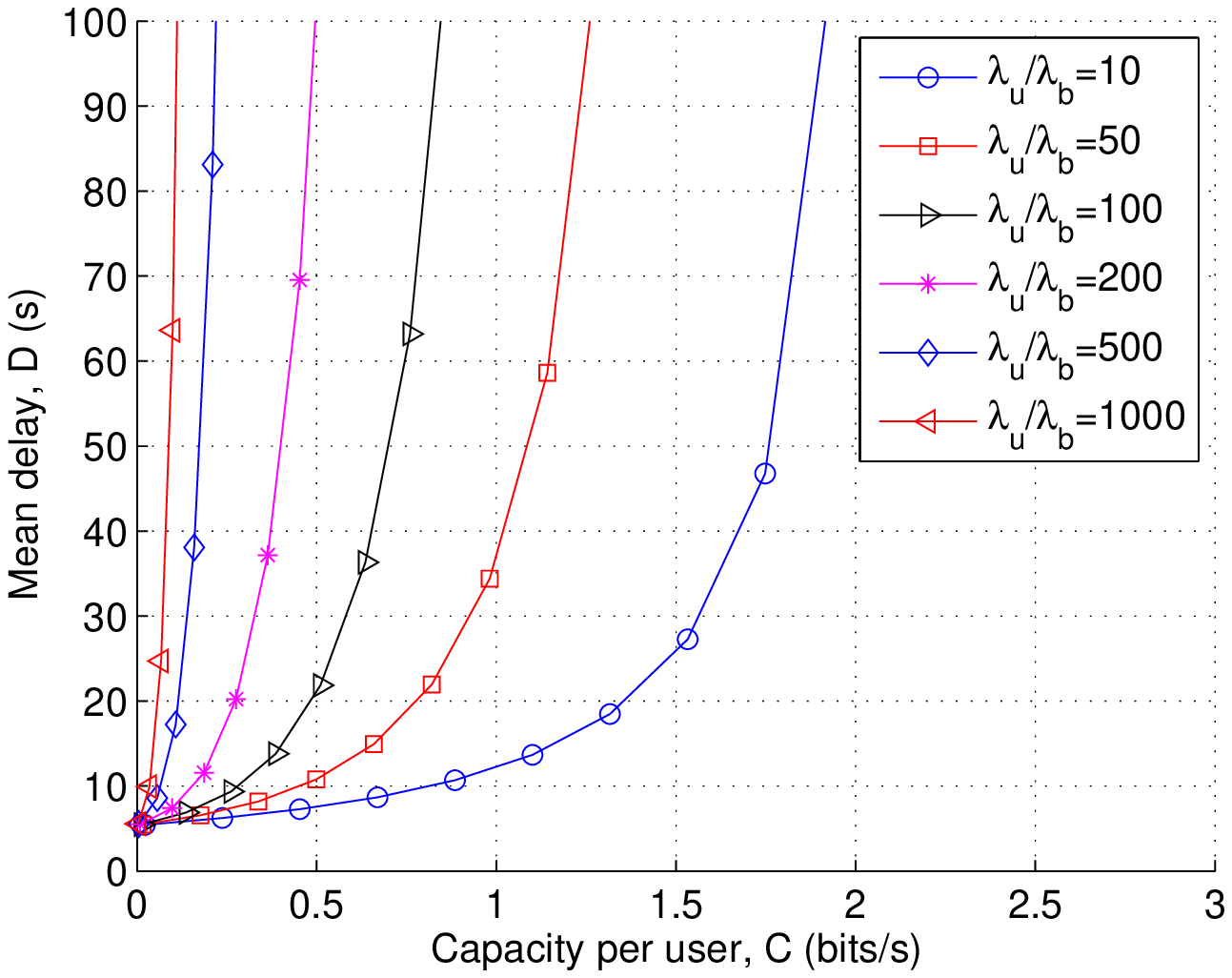}}
\caption{\small Mean delay $\bar{D}$ as a function of per user capacity $C$ with varying $\lambda_u/\lambda_b$ (fixed bandwidth per band, $N$=5, $\bar{L}$=10, $\bar{\alpha}_o$=10).}
\end{figure}

Fig.~8 shows the impact of the number of channels $N$ on the capacity-delay tradeoff curve. The capacity limits with respect to different values of $N$ are also shown. The delays are shown to rise quickly when $C$ approaches the capacity limits. It is observed that in the medium to high delay regime, capacity at a fixed delay scales linearly with $N$. In the low delay regime, increasing $N$ contributes slightly to reducing the minimum delay. Fig.~8 indicates that spectrum aggregation is effective for both capacity enhancement and delay reduction.

Fig.~9 shows the impact of average file size $\bar{L}$ on the capacity-delay tradeoff curve. The capacity limit is also shown, which is unrelated to the value of $\bar{L}$. In the low to medium capacity regime, $\bar{L}$ is shown to have a significant effect on the delay. A smaller value of $\bar{L}$ leads to a smaller delay because the file transmission has a lower probability of being interrupted by an outage. In the high delay regime, the impact of $\bar{L}$ diminishes as all delay curves eventually converge to the capacity limit. Fig.~9 suggests that file/session size management is an important factor to consider if a system is designed for low delay performance.

Fig.~10 shows the impact of mean outage arrival interval $\bar{\alpha}_o$ on the capacity-delay tradeoff curve. The capacity limit, which is independent from the values of $\bar{\alpha}_o$, is also shown. In the low delay regime, the curves converge to a minimum delay. In the high delay regime, we can predict that the curves also slowly converge to the capacity limit. However, significant differences are observed in the low to medium delay regimes. A smaller value of $\bar{\alpha}_o$ leads to smaller delays. This is because an interrupted session is less likely to be prolonged for a long period. Fig.~10 implies that introducing extra dynamics into the system (such as dynamic scheduling) can potentially help to reduce the delay.

\begin{figure}[p]
\centerline{\includegraphics[width=11cm,draft=false]{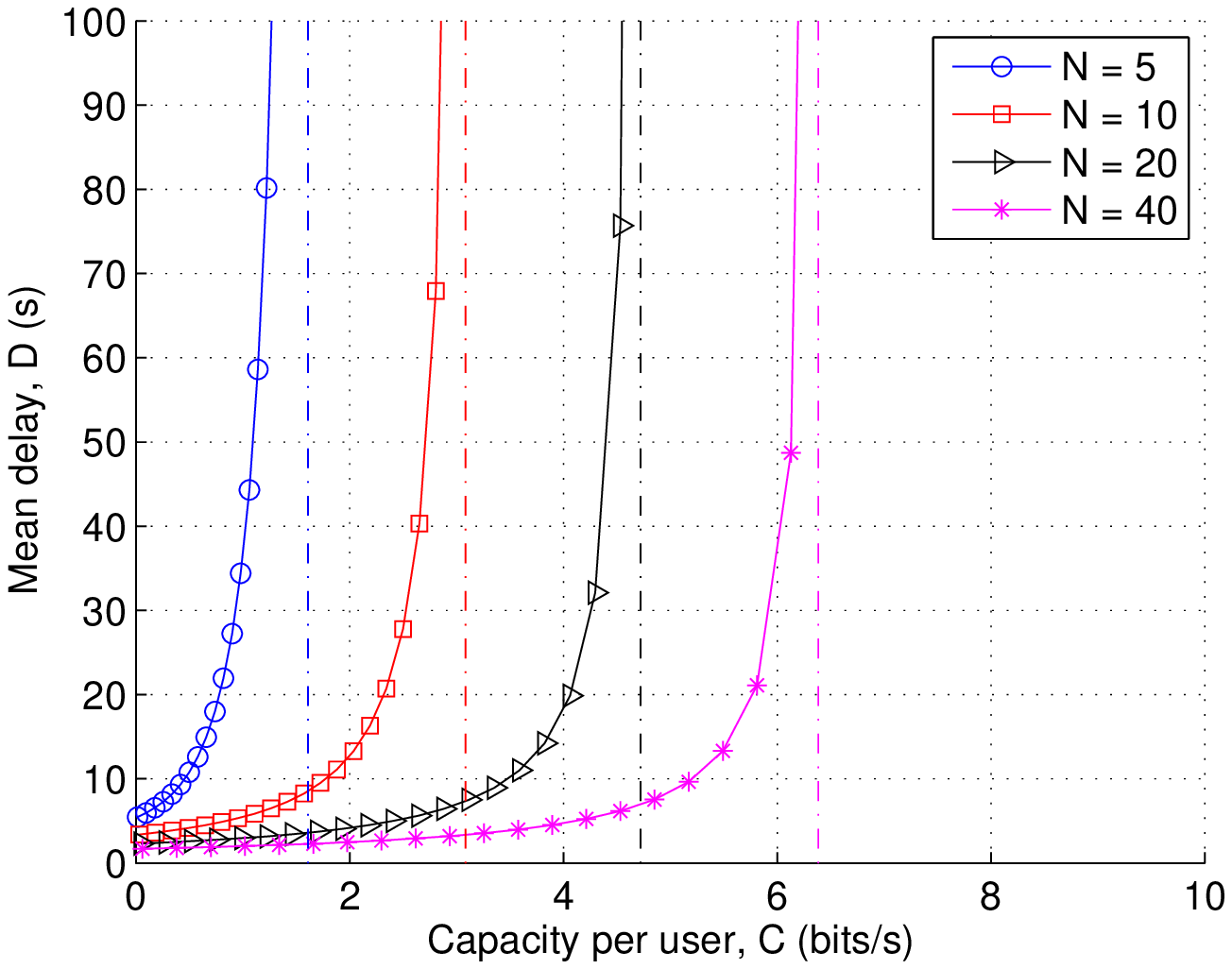}}
\caption{\small Mean delay $\bar{D}$ as a function of per user capacity $C$  with varying $N$ (fixed bandwidth per band, $\lambda_u/\lambda_b$=50, $\bar{L}$=10, $\bar{\alpha}_o$=10).}
\end{figure}
\begin{figure}[p]
\centerline{\includegraphics[width=11cm,draft=false]{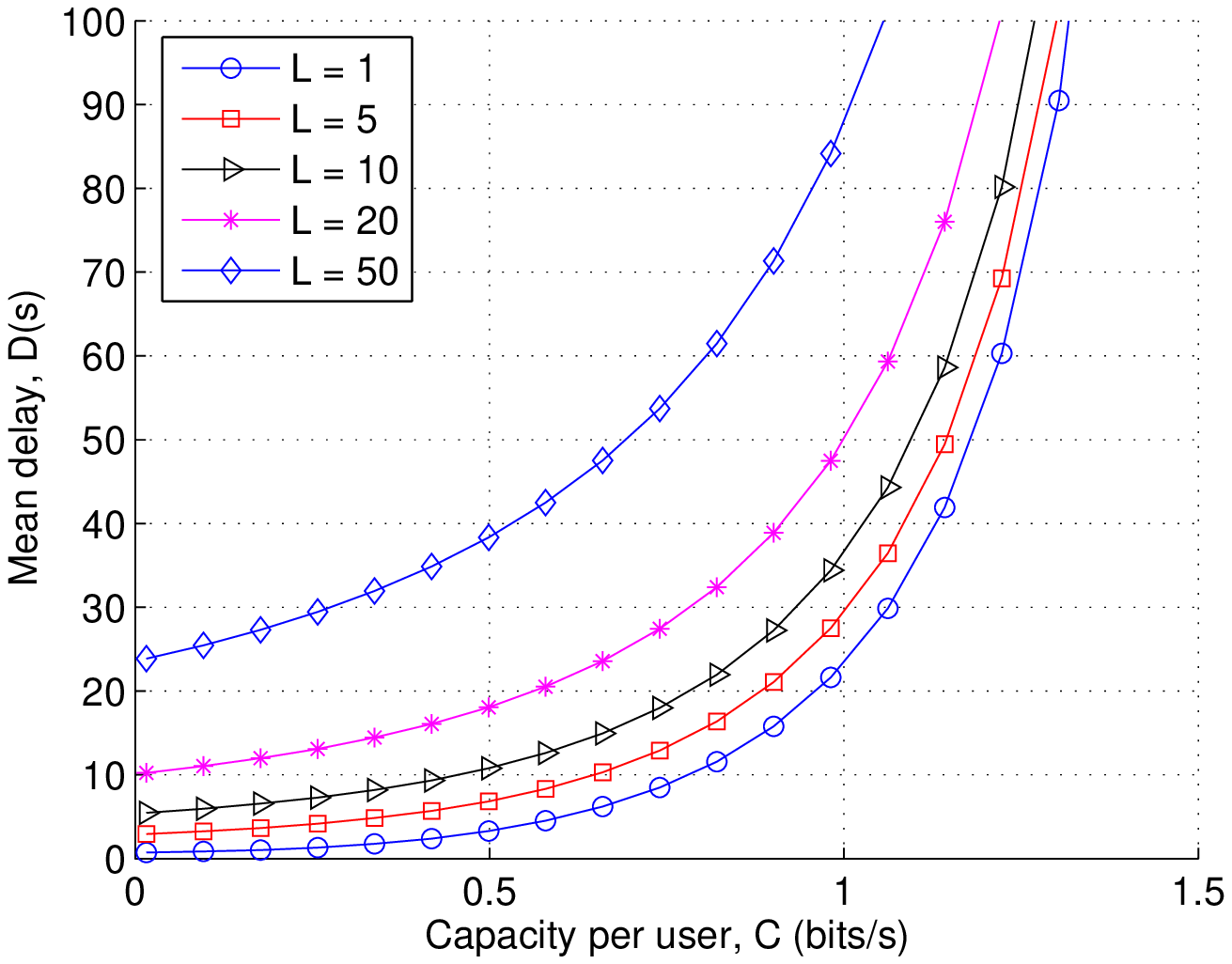}}
\caption{\small Mean delay $\bar{D}$ as a function of per user capacity $C$  with varying $\bar{L}$ (fixed bandwidth per band, $\lambda_u/\lambda_b$=50, $N$=5, $\bar{\alpha}_o$=10).}
\end{figure}

\begin{figure}[p]
\centerline{\includegraphics[width=11cm,draft=false]{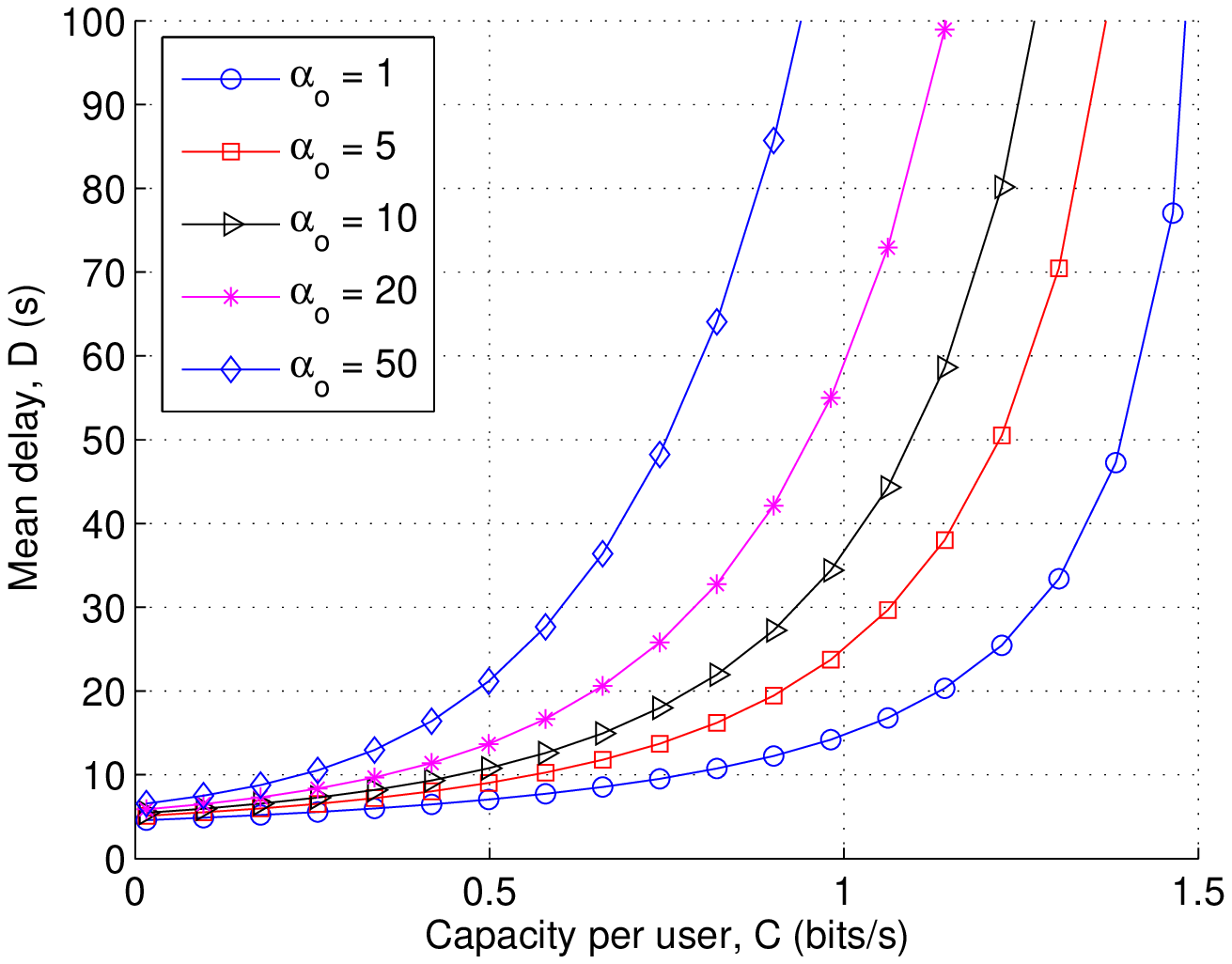}}
\caption{\small Mean delay $\bar{D}$ as a function of per user capacity $C$ with varying $\bar{\alpha}_o$ (fixed bandwidth per band, $\lambda_u/\lambda_b$=50, $N$=5, $\bar{L}$=10).}
\end{figure}

\subsection{Capacity limit and scaling}
This subsection investigates how the capacity limit scales with $N$ and user-BS density ratio.  Consider the case of fixed bandwidth per band, Fig.~11 applies Corollary 1 to show the maximum capacity $C^{\max}_I$ as a function of $N$ with varying $\lambda_u/\lambda_b$. We see that the capacity increases monotonically with increasing $N$, indicating the benefits of spectrum aggregation. However, increasing $N$ shows diminishing returns on the capacity gain. This differs from the intuition that system capacity scales linearly with bandwidth (i.e., the number of bands). It is also interesting to observe that the curves with different values of $\lambda_u/\lambda_b$ converge to the same value when $N$ tends large. This is because we assume that a user is allowed to access only one band. When $N$ is large, the capacity is limited by the bandwidth per band rather than the number of bands. Fig.~11 suggests that to achieve the full potential of spectrum aggregation, it is important to allow users to access multiple bands simultaneously.


Considering the case of fixed system bandwidth, Fig.~12 applies Proposition 5 to show the maximum capacity $C^{\max}_{II}$ as a function of $N$ with varying $\lambda_u/\lambda_b$. It is shown that with increasing $N$, the capacity increases initially but eventually declines. For each value of $\lambda_u/\lambda_b$, there exists an optimal value of $N$ to maximize the capacity. Fig.~12 reveals a design tradeoff between maximizing single channel capacity and maximizing multi-user access probability. When $N$ increases, the single channel capacity decreases due to reduced bandwidth, while the access probability increases due to increased number of bands. Fig.~12 implies that proper channelization of the available spectrum resource is important, particularly when $\lambda_u/\lambda_b$ is small.

By performing a numerical search for the optimal value of $N$ based on results in Fig.~12, Fig.~13 shows the corresponding maximum values of $C^{\max}_{II}$ as a function of $\lambda_u/\lambda_b$. We find that there exists a convenient approximation given by
\begin{equation}
C^{* \max}_{II} \approx 0.6359-0.052\log_2(\lambda_u/\lambda_b). \label{ciiappx}
\end{equation}
The actual values obtained from numerical calculation and the approximated values obtained from (\ref{ciiappx}) are compared in Fig.~13. It is shown that our approximation is reasonably accurate for $2 < \lambda_u/\lambda_b < 500$. Fig.~13 shows that per user throughput is upper bounded by a constant and reduces at a sub-linear rate with increasing $\lambda_u/\lambda_b$.

\begin{figure}[p]
\vspace{0.1in}
\centerline{\includegraphics[width=11cm,draft=false]{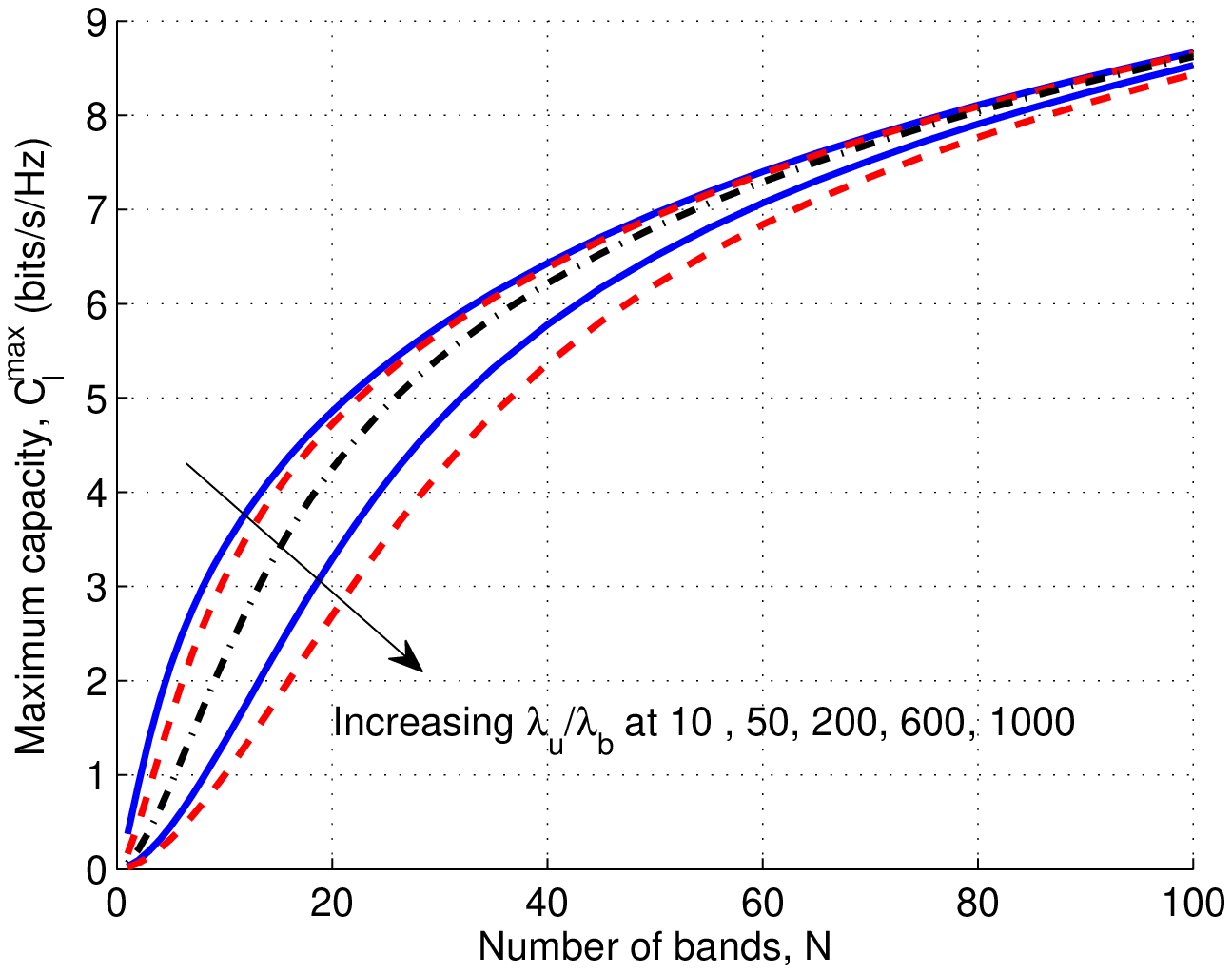}}
\caption{\small Maximum capacity $C^{\max}_I$ as a function of $N$ with varying $\lambda_u/\lambda_b$ (fixed bandwidth per band).}
\end{figure}

\begin{figure}[p]
\vspace{0.125in}
\centerline{\includegraphics[width=11cm,draft=false]{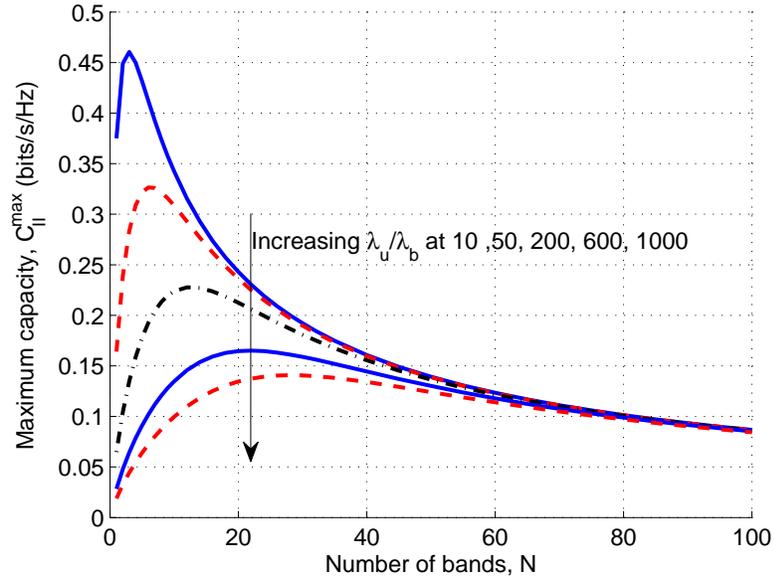}}
\caption{\small Maximum capacity $C^{\max}_{II}$ as a function of $N$ with varying $\lambda_u/\lambda_b$ (fixed system bandwidth).}
\end{figure}

\begin{figure}[p]
\vspace{0.125in}
\centerline{\includegraphics[width=11cm,draft=false]{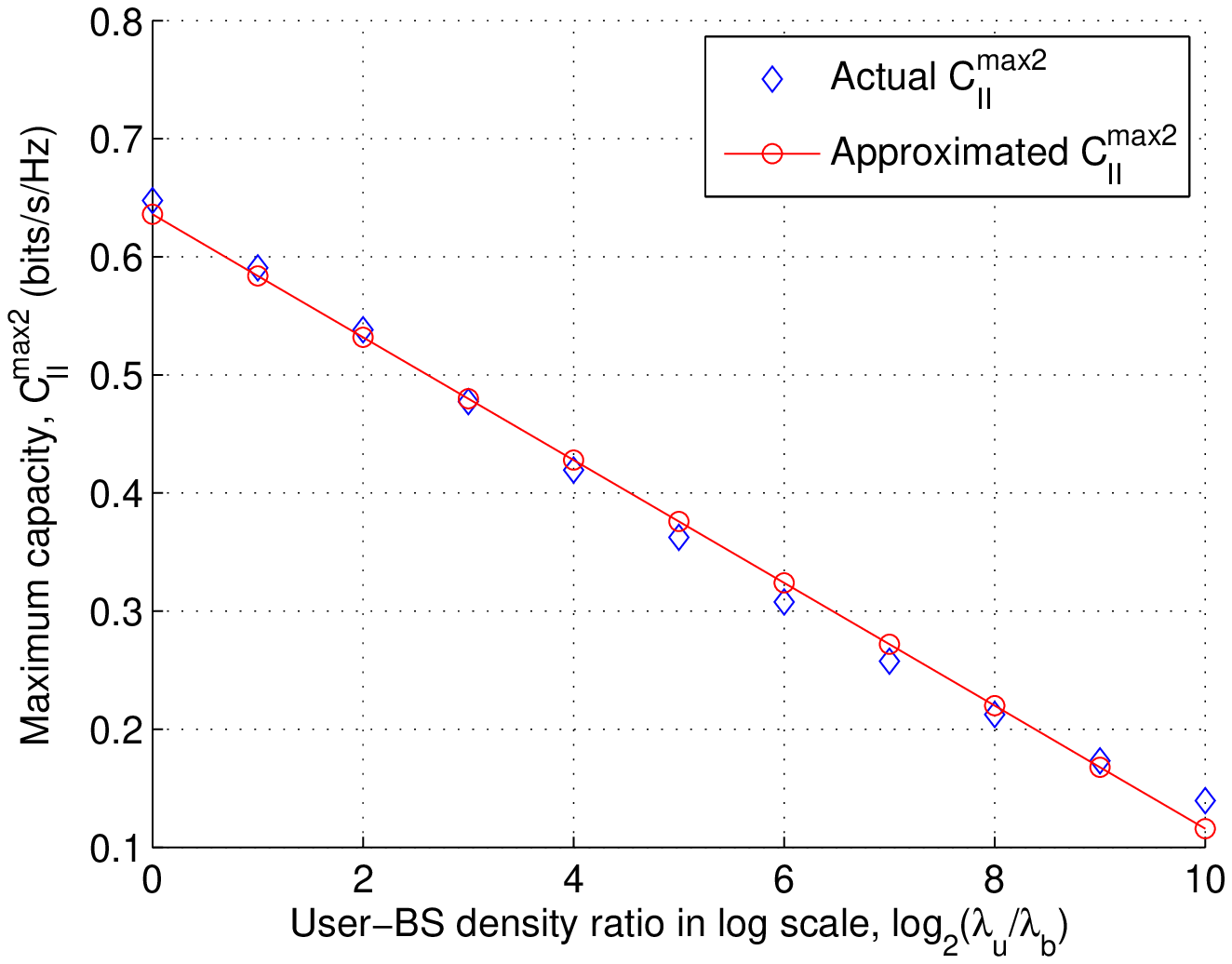}}
\caption{\small Approximation of $C^{\max}_{II}$ as a function of user-BS density ratio (fixed system bandwidth).}
\end{figure}

\section{Conclusions}
An analytical framework has been proposed for the study of the capacity-delay tradeoff in cellular networks with spectrum aggregation. The framework compliments existing ones by focusing on the secondary traffic and offering tractable analytical insights. Analytical results have been derived to characterize the capacity-delay tradeoff and the fundamental capacity limit. Numerical studies have shown that while spectrum aggregation primarily affects the capacity in the high-delay regime, session size management and dynamic scheduling have bigger impacts on the capacity in the low delay regime. Moreover, when different bands have homogeneous configurations, it has been shown that the per user throughput per Hertz is upper bounded by a constant and reduces at a rate proportional to the logarithm of user-BS density ratio. Our analysis offers useful guidelines for providing novel secondary services over cellular networks to improve the overall capacity utilization.

\appendices
\section{}
This appendix gives the proof for Lemma 1. Because we assume that an active user randomly selects a band for access, in an equilibrium state, the density of users in a band is proportional to the area fraction of coverage of this band. The density of active users in the $n$th band is then given by
\begin{equation}
\lambda_{u,n} = \lambda_u \cdot p_{active} \cdot \frac{\Omega_n p_{n,v}}{\sum_{n=1}^N \Omega_n p_{n,v}} = \frac{\lambda_u \rho_s}{\varepsilon} \cdot \frac{\Omega_n p_{n,v}}{\sum_{n=1}^N \Omega_n p_{n,v}}.
\end{equation}
 Now consider an active user in band $n$, the number of contenting users in the same cell can be evaluated according to (\ref{fK}) with user density $\lambda_{u,n}$ and BS density $\lambda_{b,n}$. When strict fairness is assumed, the access probability of a user is given by
\begin{align}\label{pna}
 p_{n,a} =& \sum_{k=0}^{\infty} \frac{1}{k+1} f_K(k) \\  \nonumber
         =& \sum_{k=0}^{\infty} \frac{1}{k+1} \int_{k=0}^{\infty} \frac{(\lambda_n \Lambda p_{n,v} x)^k}{k!} e^{-\lambda_n \Lambda p_{n,v} x} f_U(x) dx \\ \nonumber
         =& \sum_{k=0}^{\infty} \frac{1}{\lambda_n \Lambda p_{n,v} x} \left[ \int_{k=1}^{\infty} \frac{(\lambda_n \Lambda p_{n,v} x)^k}{k!} \right] e^{-\lambda_n \Lambda p_{n,v} x} f_U(x) dx \\ \nonumber
         =& \sum_{k=0}^{\infty} \frac{1}{\lambda_n \Lambda p_{n,v} x} \left( 1- e^{-\lambda_n \Lambda p_{n,v} x}\right) f_U(x) dx \\ \nonumber
         =& \frac{3.5^{4.5}}{\Gamma(4.5)} \frac{1}{\lambda_n \Lambda p_{n,v}} \left[ \int_0^{\infty} x^{2.5} e^{-3.5x} dx - \int_0^{\infty} x^{2.5} e^{(-3.5 + \lambda_n \Lambda p_{n,v})x} dx   \right] \\ \nonumber
         = & \frac{3.5^{4.5}}{\Gamma(4.5)} \frac{1}{\lambda_n \Lambda p_{n,v}}\left[ \frac{\Gamma(3.5)}{3.5^{3.5}} - \frac{\Gamma(3.5)}{(3.5 + \lambda_n \Lambda p_{n,v})^{3.5}}   \right] \\ \nonumber
         =& \frac{1}{\lambda_n \Lambda p_{n,v}} \left[ 1- \left( 1 + \frac{\Lambda p_{n,v} \lambda_n}{3.5} \right)^{-3.5} \right].
\end{align}
Finally, Lemma 1 can be obtained by substituting (\ref{pna}) into (\ref{epin}).


\end{document}